\documentclass[a4paper,UKenglish]{lipics-v2021}

\usepackage[T1]{fontenc}%for Leszek' name
\usepackage{amsfonts,amsmath,amssymb,amsthm}
\usepackage{tcolorbox}
\usepackage{multicol}
\usepackage{environ}
\usepackage[symbol]{footmisc}
\usepackage{hyperref}

\usepackage{orcidlink}

\newtheorem{fact}{Fact}

\usepackage{graphicx} % Required for inserting images
\usepackage{xcolor}   % for comments, while work in progress
\usepackage[ruled]{algorithm2e}

\usepackage{xparse} % for handling optional arguments in new defined commands

\usepackage{comment}

\nolinenumbers

\newcommand{\lkcomment}[1]{\textcolor{blue}{LK: #1}}
\newcommand{\gst}[1]{\textcolor{red}{#1}}

\newcommand{\initiator}{i}

\newcommand{\responder}{r}

\newcommand{\distance}[2]{d_{{#1}{#2}}}

\newcommand{\agentsymbol}[1]{
{A_{#1}}
}

\newcommand{\agentlabel}[1]{
{x_{#1}}
}

\newcommand{\leader}{L}

\newcommand{\agentcolor}[1]{
    {c_{#1}}
}
\newcommand{\greystate}{0}
\newcommand{\greenishstate}{1}
\newcommand{\greenstate}{2}

\NewDocumentCommand{\agentvlabel}{m o}{
\IfNoValueTF{#2}
    {\mathbf{x_{#1}}}
    {\mathbf{x_{#1}[#2]}}
}

\newcommand{\agentvector}[2]{\overrightarrow{v_{{#1}{#2}}}}

\NewDocumentCommand{\agentvvector}{m m o}{
\IfNoValueTF{#3}
    {\overrightarrow{\mathbf{v_{{#1}{#2}}}}}
    {\overrightarrow{\mathbf{v_{{#1}{#2}}}}[#3]}
}

\newcommand{\agentposition}[1]{
{p_{#1}}
}

\newcommand{\algoassign}{\leftarrow}

\newcommand{\commontcc}{Code executed by initiator $\agentsymbol{i}$ interacting with ${\agentsymbol{r}}$}

\SetKw{Or}{\, or \,}
\SetKw{And}{\, and \,}
\SetKwInOut{Input}{received data}
\SetKwInOut{Local}{local state}
\SetKwInOut{Output}{output}
\SetKwInOut{Parameters}{parameters}

\title{Anonymous Self-Stabilising Localisation via Spatial Population Protocols}
\titlerunning{Anonymous Self-Stabilising Localisation via Spatial Population Protocols}

\author{Leszek Gąsieniec}{Department of Computer Science, University of Liverpool, UK. \and \url{https://cgi.csc.liv.ac.uk/~leszek/}}{L.A.Gasieniec@liverpool.ac.uk}{https://orcid.org/0000-0003-1809-9814}{}
\author{Łukasz Kuszner}{Institute of Informatics, University of Gda{\'n}sk, Poland \and \url{https://inf.ug.edu.pl/~lkuszner/}}
{lukasz.kuszner@ug.edu.pl}{https://orcid.org/0000-0003-1902-7580}{}

\author{Ehsan Latif}{AI4STEM Education Center, University of Georgia, Athens, GA, USA \and \url{https://ehsanlatif.github.io/}}{ehsan.latif@uga.edu}{https://orcid.org/0009-0008-6553-4093}{}

\author{Ramviyas Parasuraman}{School of Computing, University of Georgia, Athens, GA, USA \and \url{https://computing.uga.edu/directory/people/ramviyas-nattanmai-parasuraman}}{ramviyas@uga.edu}{https://orcid.org/0000-0001-5451-8209}{}

\author{Paul Spirakis}{Department of Computer Science, University of Liverpool, UK. \and \url{https://www.liverpool.ac.uk/people/paul-spirakis}} 
{P.Spirakis@liverpool.ac.uk}{https://orcid.org/0000-0001-5396-3749}{Supported by the EPSRC grant EP/P02002X/1.}

\author{Grzegorz Stachowiak}{Institute of Computer Science, University of Wroc{\l}aw, Poland \and \url{https://ii.uni.wroc.pl/~gst/}}
{Grzegorz.Stachowiak@ii.uni.wroc.pl}{https://orcid.org/0000-0003-0463-3676}{Supported by the NCN grant 2020/39/B/ST6/03288.}


\authorrunning{L. Gąsieniec, {Ł}. Kuszner, E. Latif, R. Parasuraman, P. Spirakis, G. Stachowiak}
\Copyright{Leszek Gąsieniec, Łukasz Kuszner, Ehsan Latif, Ramviyas Parasuraman, Paul Spirakis, Grzegorz Stachowiak}

\begin{document}
\maketitle

\begin{abstract}
In the \textit{distributed localization problem} (DLP), \(n\) anonymous robots (agents) \(\agentsymbol{0}, \dots, \agentsymbol{n-1}\) begin at arbitrary positions \(\agentposition{0}, \dots, \agentposition{n-1} \in S\), where \(S\) is a Euclidean space. Initially, each agent \(\agentsymbol{i}\) operates within its own coordinate system in \(S\), which may be inconsistent with those of other agents. The primary goal in DLP is for agents to reach a consensus on a unified coordinate system that accurately reflects the relative positions of all points, \(\agentposition{0}, \dots, \agentposition{n-1}\), in \(S\).
Research on DLP has largely focused on the feasibility and complexity of achieving consensus with limited inter-agent distance data, often due to missing or imprecise information. In contrast, this paper explores a minimalist, computationally efficient distributed computing model where agents can access all pairwise distances if required.
%Extensive research on DLP has primarily focused on the feasibility and complexity of achieving consensus when agents have limited access to inter-agent distances, often due to missing or imprecise data.
%
%In this paper, however, we examine a minimalist, computationally efficient model of distributed computing in which agents have access to all pairwise distances, if needed.
%
Specifically, we introduce a novel variant of population protocols, referred to as the \textit{spatial population protocols} model. In this variant each agent can memorise one or a fixed number of coordinates, and when agents \(\agentsymbol{i}\) and \(\agentsymbol{j}\) interact, they can not only exchange their current knowledge but also either determine the distance \(\distance{i}{j}\) between them in \(S\) (distance query model) or obtain the vector \(\agentvector{i}{j}\) spanning points \(\agentposition{i}\) and \(\agentposition{j}\) (vector query model).

We propose and analyse several distributed localisation protocols, including:
%We examine three DLP scenarios, proposing and analysing several types of distributed localisation protocols, including:
%the initial configuration and the information gathered during an interaction between two anonymous agents, $\agentsymbol{i}$ and $\agentsymbol{j}$. 
%
\begin{enumerate}
%\item {\em Self-stabilising localisation protocol with distance queries}
%We propose and analyse self-stabilising localisation protocol based on pairwise distance adjustment. We also discuss several hard instances in this scenario, and suggest possible improvements for the considered protocol,

%To overcome the complexity issues including stabilisation we strengthen this model in two different ways.
%
\item {\em Leader-based localisation protocol with distance queries} 
We propose and analyse two leader-based localisation protocols that stabilize silently in $o(n)$ time. These protocols leverage an efficient solution to the novel concept of {\em multi-contact epidemic}, a natural generalisation of the core communication tool in population protocols, known as the {\em one-way epidemic}.
%We propose and analyse two leader-based localisation protocols which stabilise in $o(n)$ parallel time. These protocols rely on efficient solution to {\em multi-contact epidemic}, a natural generalisation of the central communication tool in population protocols {\em one-way epidemic}, and

\item {\em Self-stabilising leader localisation protocol with distance queries}
We show how to effectively utilise a leader election mechanism within the leader-based localisation protocol to get DLP protocol that self-stabilises silently in time $O(n(\log n/n)^{1/(k+1)}\log n)$ in $k$-dimensions. 

\item {\em Self-stabilising localisation protocol with vector queries} 
We propose and analyse an optimally fast DLP protocol which self-stabilises silently in $O(\log n)$ time.
\end{enumerate}

We conclude by outlining future research directions for distributed localisation in spatial population protocols, covering higher dimensions, limited precision, and error susceptibility.

%We conclude with a discussion on future research directions for distributed localisation in spatial population protocols, including scenarios that account for higher dimensions, limited precision and susceptibility to errors.

%A brief announcement with partial results of this paper will appear in SAND 2025.
\end{abstract}

\section{Introduction}

Location services are crucial for modern computing paradigms, such as pervasive computing and sensor networks. While manual configuration and GPS can determine node locations, these methods are impractical in large-scale or obstructed environments. Recent approaches use network localisation, where beacon nodes with known positions enable other nodes to estimate their locations via distance measurements. Key challenges remain, including determining conditions for unique localisability, computational complexity, and deployment considerations.
%
%Aspnes et al.~\cite{AspnesEGMWYAB06} address these issues through graph rigidity theory, introducing grounded graphs to improve localisation efficiency and exploring the density-based complexity of node localisation in sensor networks.
%
%
In the \textit{distributed localisation problem} (DLP), \(n\) anonymous robots (agents) \(\agentsymbol{0}, \dots, \agentsymbol{n-1}\) begin at arbitrary positions \(\agentposition{0}, \dots, \agentposition{n-1} \in S\), where \(S\) is a Euclidean space. Initially, each agent \(\agentsymbol{i}\) operates within its own coordinate system in \(S\), which may be inconsistent with those of other agents. The primary goal in DLP is for agents to reach a consensus on a unified coordinate system that accurately reflects the relative positions of all points, \(\agentposition{0}, \dots, \agentposition{n-1}\), in \(S\).

A network of agents' unique localisability is determined by specific combinatorial properties of its graph and the number of {\em anchors} (agents aware of their real location). 
For example, {\em graph rigidity theory}~\cite{Eren, Hendrickson, Jackson, LNP+23} provides a necessary and sufficient condition for unique localisability~\cite{Eren}. Specifically, a network of agents located in the plane is uniquely localisable if and only if it has at least three anchors and the network graph is {\em globally rigid}. However, unless a network is dense and regular, global rigidity is unlikely. Even without global rigidity, large portions of a network may still be globally rigid, though positions of remaining nodes will remain indeterminate due to multiple feasible solutions.
The decision version of this problem, often referred to as the {\em graph embedding} or {\em graph realisation} problem, requires determining whether a weighted graph can be embedded in the plane so that distances between adjacent vertices match the edge weights, a problem known to be strongly NP-hard~\cite{Saxe}. Furthermore, this complexity holds even when the graph is globally rigid~\cite{Eren}.
In {\em sensor networks}, where nodes measure distances only within a range rr, the network is naturally modeled as a {\em unit disk graph}, with two nodes adjacent if and only if their distance is $\le r$. The corresponding decision problem, {\em unit disk graph reconstruction}, asks whether a graph can be embedded in the plane so that adjacent nodes are within distance $r$, and non-adjacent nodes are farther apart.
This problem is also NP-hard~\cite{AspnesGY04}, indicating that no efficient algorithm can solve localisation in the worst case unless \(P = NP\). Furthermore, even for instances with unique reconstructions, no efficient randomised algorithm exists to solve this problem unless \(P = NP\)~\cite{AspnesGY04}.
Although several distributed computing models for mobile agents with properties similar to population protocols have been proposed, see the comprehensive survey~\cite{FPS19}, the localisation problem has not been studied in any of them.

Distributed localisation is also crucial in {\em robotic systems}, enabling robots to autonomously determine their spatial position within an environment  -- a fundamental requirement for applications such as navigation, mapping, and multi-robot coordination \cite{yue2020collaborative}. Accurate localisation allows robots to interact more effectively with their surroundings and with each other, facilitating tasks from autonomous driving to warehouse automation and search-and-rescue operations \cite{latif2023seal, misaros2023autonomous}. Localisation approaches generally fall into two broad categories: centralised and distributed systems \cite{zafari2019survey}. Centralised localisation systems, where a central server or leader node computes the locations of all robots, can offer high accuracy but may struggle with scalability and robustness, especially in dynamic or communication-limited environments \cite{latif2023instantaneous}. In contrast, distributed localisation systems allow each robot to perform localisation computations independently or in collaboration with neighbouring robots, enhancing adaptability and resilience, although this may come at the cost of increased complexity \cite{latif2022dgorl, latif2022multi}. Within distributed systems, leader-based localisation mechanisms involve one or more designated robots that serve as reference points or coordinators for localisation \cite{fareh2023logarithmic}, which can streamline computations but may create single points of failure. Leaderless localisation, where all robots contribute equally to position estimation without relying on specific leader nodes, is advantageous in decentralised applications where flexibility and fault tolerance are paramount \cite{latif2022dgorl, wang2023distributed}. Both methods have been explored using probabilistic \cite{xu2020probabilistic}, geometric \cite{latif2023gprl}, and graph-based models \cite{latif2022dgorl}, with leaderless approaches gaining traction due to their robustness in large-scale and dynamically changing settings. Various methods leverage tools such as Kalman filters \cite{moreira2020new}, particle filters \cite{latif2023instantaneous}, and graph rigidity theory \cite{guo2021semantic} to enhance localisation accuracy and efficiency in complex environments.

Another variant of the localization problem, known as the \( n \)-{\em point location problem}, has been studied in a centralized setting~\cite{Dam03,Dam06}. This problem was motivated by and is closely related to the computation of relative positions of markers on a DNA string~\cite{HIL03,RR00,Mum00}.  
In the \( n \)-{\em point location problem}, the objective is to design one or more rounds of pairwise distance queries between points to determine the relative distances among all points.  
In~\cite{Dam03}, a one-round deterministic strategy is presented that utilizes \( \frac{8n}{5} \) distance queries, along with a proof that any one-round strategy requires at least \( \frac{4n}{3} \) distance queries. The same work also introduces a simple two-round deterministic strategy that uses \( \frac{3n}{2} \) queries.  
An alternative two-round randomised approach, which requires only \( n + O\left(\frac{n}{\log n}\right) \) queries and solves the \( n \)-{\em point location problem} with high probability, is described in~\cite{Dam06}.

\subsection{Spatial population protocols\label{SPP}}
In this paper, we explore a minimalist, computationally efficient model of distributed computing, where agents have probabilistic access to pairwise distances. Our focus is on achieving {\em anonymity} while maintaining high time efficiency and minimal use of network resources, including limited local storage (agent state space) and communication.
To meet these goals, we introduce a new variant of population protocols, referred to as the \textit{spatial population protocols} model, specified later in this section.

The population protocol model originates from the seminal work by Angluin et al.~\cite{DBLP:conf/podc/AngluinADFP04}. This model provides tools for the formal analysis of \textit{pairwise interactions} between indistinguishable entities known as \textit{agents}, which have limited storage, communication, and computational capabilities.
When two agents interact, their states change according to a predefined \textit{transition function}, a core component of the population protocol. It is typically assumed that the agents' \textit{state space} is fixed and that the population size \(n\) is unknown to the agents and is not hard-coded into the transition function.
In \textit{self-stabilising} protocols, the \textit{initial configuration} of agents' states is arbitrary. By contrast, non-self-stabilising protocols start with a predefined configuration encoding the input of the given problem. A protocol concludes when it \textit{stabilises} in a \textit{final configuration} of states representing the solution.
In the \textit{probabilistic variant} of population protocols, which is also used here, the \textit{random scheduler} selects an ordered pair of agents at each step—designated as the \textit{initiator} and the \textit{responder}—uniformly at random from the entire population. The asymmetry in this pair introduces a valuable source of random bits, which is utilised by population protocols.
In this probabilistic setting, besides \textit{efficient state utilisation}, \textit{time complexity} is also a primary concern. It is often measured by the number of interactions, \({\cal I}\), required for the protocol to stabilise in a final configuration. More recently, the focus has shifted to \textit{parallel stabilisation time} (or simply \textit{time}), defined as \({\cal I}/n\), where \(n\) is the population size. This measure captures the parallelism of independent, simultaneous interactions, which is leveraged in \textit{efficient population protocols} that stabilise in time \(O(\text{poly} \log n)\).
All protocols introduced in this paper are {\em stable} (consistently correct), {\em stabilise silently} (agent states cease to change post-stabilization), and guarantee stabilization time with high probability (whp), defined as $1 - n^{-\eta}$ for a constant $\eta > 0$. Furthermore, our primary contribution, the protocols detailed in Sections~\ref{s:self} and \ref{s:vector}, exhibits self-stabilization.

Among the primary tools in our localization protocols, we emphasize the novel concept of {\em multi-contact epidemic}, see Section~\ref{s:leader}, a natural extension of the fundamental communication mechanism in population protocols, known as the {\em one-way epidemic}. Another key tool is {\em leader election}, a cornerstone of distributed computing that tackles the essential challenges of symmetry breaking, synchronization, and coordination. In population protocols, the presence of a leader facilitates a more efficient computational framework~\cite{AngluinAE08}. However, achieving leader election within this model presents significant difficulties.
Foundational results~\cite{DBLP:conf/wdag/ChenCDS14, DBLP:conf/soda/Doty14} demonstrate that it cannot be solved in sublinear time if agents are restricted to a fixed number of states~\cite{DBLP:journals/dc/DotyS18}.
Further, Alistarh and Gelashvili~\cite{DBLP:conf/icalp/AlistarhG15} introduced a protocol stabilising in \(O(\log^3 n)\) time with \(O(\log^3 n)\) states. Later, Alistarh \textit{et al.}~\cite{DBLP:conf/soda/AlistarhAEGR17} identified trade-offs between state use and stabilisation time, distinguishing \textit{slowly} (\(o(\log \log n)\) states) and \textit{rapidly stabilising} (\(O(\log n)\) states) protocols.
Subsequent work achieved \(O(\log^2 n)\) time whp and in expectation with \(O(\log^2 n)\) states~\cite{DBLP:conf/podc/BilkeCER17}, later reduced to \(O(\log n)\) states using synthetic coins~\cite{DBLP:conf/soda/AlistarhAG18, DBLP:conf/soda/BerenbrinkKKO18}. Recent research by G\k asieniec and Stachowiak reduced state usage to \(O(\log \log n)\) while retaining \(O(\log^2 n)\) time whp~\cite{DBLP:journals/jacm/GasieniecS21}. The expected time of leader election was further optimised to $O(\log n\log\log n)$ by G\k asieniec {\em et al.} in \cite{DBLP:conf/spaa/GasieniecSU19} and to the optimal time \(O(\log n)\) by Berenbrink {\em et al.} in~\cite{DBLP:conf/stoc/BerenbrinkGK20}.
In contrast, self-stabilising leader election protocols present unique computational challenges. Notably, it was shown by Cai {\em et al.} in~\cite{CIW12} that such protocols require at least $n$ states, in addition to knowledge of the exact value of $n$.
%, a requirement similar to the ranking problem. 
Alternatively, if loose-stabilisation is allowed (where a leader remains in place for a long but finite time before re-election) an upper bound on $n$ may suffice, see work of Sudo {\em et al.}~\cite{Sudo+20}.
In~\cite{BCC+21,DS18}, it was established that any silently self-stabilising leader election protocol has an expected time complexity of at least $\Omega(n)$. Furthermore, Burman {\em et al.} in~\cite{BCC+21} present silently self-stabilising protocols with expected time $O(n)$ and with high probability (whp) time $O(n \log n)$, both using $n + \Omega(n)$ states. More recently, the space complexity of $O(n \log n)$-time leader election self-stabilising whp protocol has been improved to $n + O(\log^2 n)$ in~\cite{BEG+25} and
to $n + O(\log n)$ in~\cite{GGS25}.
In this work, however, we employ a relatively straightforward leader election mechanism where each agent collects $\log n$ random bits, and a leader is determined by a complete set of $\log n$ 1s. This simple protocol ensures the election of a unique leader with constant probability. When repeated $O(\log n)$ times, this procedure guarantees unique leader election whp. For details see Section~\ref{s:self}.

\paragraph*{Spatial embedding and geometric queries}
While population protocols provide an elegant and resilient framework for randomised distributed computation, they tend to lack spatial embedding. To address this limitation, we introduce a novel {\em spatial} variant of population protocols that extends the transition function to include basic geometric queries. 
In particular, in this model each agent
can memorise one or a fixed number of coordinates, and during an interaction of two agents 
$\agentsymbol{i}$ and $\agentsymbol{j}$, in addition to exchange of their current knowledge, the agents can determine:
\begin{description}
    \item{(1)} the distance $\distance{i}{j}$ separating them in $S$, in {\em distance query model}, and
    \item{(2)} vector $\agentvector{i}{j}$ spanning points $\agentposition{i}$ and $\agentposition{j}$, in {\em vector query model}.
\end{description}
%

%\subsection{Our contribution}
%\paragraph{\bf Our contribution}
%\lkcomment{to be written later}

\subsection{Importance and structure of presentation}

As previously discussed, the localisation problem is fundamental, with its many variants extensively studied over time.
This work effectively integrates research in (distributed) localisation with rigidity theory of random geometric graphs, leveraging the increasingly popular population protocol model of computation.
Surprisingly, little research has explored population protocols incorporating spatial properties, particularly those supporting geometric queries. The most relevant prior work focuses on efficient majority and leader election protocols for general toroidal grids~\cite{AGR22} and self-stabilizing leader election in rings~\cite{YSM21,YSO+23}.

In this study, we assume agents are distributed in a $k$-dimensional space, where each interaction, viewed as a pairwise geometric query, results in an exchange of states (knowledge held by agents) augmented by information about the relative positions of the interacting agents. 
We present the following results.
%
%\item {\em Self-stabilising localisation protocol with distance queries}
%We propose and analyse self-stabilising localisation protocol based on pairwise distance adjustment. We also discuss several hard instances in this scenario, and suggest possible improvements for the considered protocol,

%To overcome the complexity issues including stabilisation we strengthen this model in two different ways.
%
In Section~\ref{s:leader}, we propose and analyse two leader-based localisation protocols in the model with distance queries.
The first protocol (Algorithms~\ref{alg:Pos_k}) stabilises silently in time $O(n(\log n/n)^{1/(k+1)}$ whp (Theorem~\ref{t:PosK}), in $k$-dimensional space. 
For $k=1,$ we present a faster protocol with the stabilisation time \(O((n\log n)^{1/3})\) whp (Theorem~\ref{t:improved}).
These two protocols leverage an efficient solution to the novel concept of {\em multi-contact epidemic} (Section~\ref{k-con}), a natural generalisation of the core communication tool in population protocols, known as the {\em one-way epidemic}.
%We propose and analyse two leader-based localisation protocols which stabilise in $o(n)$ parallel time. These protocols rely on efficient solution to {\em multi-contact epidemic}, a natural generalisation of the central communication tool in population protocols {\em one-way epidemic}, and
%
In Section~\ref{s:self},
we show how to effectively utilise leader election within the leader-based localisation protocol (Algorithm~\ref{alg:Pos_k}) in the model with distance queries.
We propose and analyse a DLP protocol that self-stabilises silently in time $O(n(\log n/n)^{1/(k+1)}\log n)$ whp (Theorem~\ref{t:self}), in $k$-dimensional space. 
Finally, in Section~\ref{s:vector}, we present an optimally fast DLP protocol which self-stabilises silently in $O(\log n)$ time whp (Theorem~\ref{t:super}), in the model with vector queries and fixed $k.$

We conclude with a discussion on future research directions for distributed localisation in spatial population protocols, including scenarios that account for higher dimensions, limited precision, and susceptibility to errors.

\section{Leader-based localisation in distance query model \label{s:leader}}

In this section, we discuss two localisation protocols with predefined input stabilising in $o(n)$ time, i.e., after this time labels of all agents become stable. These protocols are non-self-stabilising. We assume that one of the agents starts as the \emph{leader} of the population. If the identity of the leader is not known, the localisation protocol can be preceded by one of the leader election protocols discussed in the introduction.
The agents' positions \(p_0, \dots, p_{n-1}\) are distributed in a \(k\)-dimensional Euclidean space \(S\). %where \(k\) is a fixed integer. 
It is assumed that any \(k+1\) agents' positions span the entire space. For example, in two-dimensional space, this assumption guarantees that no three points are collinear.
Although our algorithms can be adapted to handle an arbitrary distribution of agents' positions, the time guarantees of such adaptations would be weaker. The state of an agent can accommodate a fixed (related to $k$) number of agent positions and distances.

We adopt a symmetric model of communication, which means that when agents $A_u$ and $A_v$ interact, they both gain access to each other's states as well as the distance $d_{uv}$. The transitions assigned to the leader are distinct from those of the other agents, and the leader also serves as the initiator of the entire process.
Initially, the state of each agent \(A_u\) stores a label \(x_u\) (representing a hypothetical position in \(S\)) and its colour \(C[A_u]\). We assume that at the beginning, the leader is coloured \texttt{green}, and each non-leader's colour is set to \texttt{blue}.
Finally, the leader's label (position in $S$) is set at the origin of the coordinate system, i.e., this label is used as the anchor in the localisation process.

\subsection{Multi-contact epidemic\label{k-con}}

We introduce and analyse the process of $k$-\emph{contact epidemic} (a multi-contact epidemic with a fixed parameter $k$), a natural generalization of epidemic dynamics in population protocols. In this process, the population initially contains at least $k$ \texttt{green} agents, while the remaining agents are \texttt{blue}. A \texttt{blue} agent turns \texttt{green} after interacting with $k$ distinct \texttt{green} agents. We demonstrate that the time complexity of this process is  \(O(n^{1 - 1/k} \log n)\), for any fixed integer $k$.  
%To establish this result, 

We begin with two key lemmas to establish this result.

\begin{lemma}\label{double_k}
The time needed to transition from a configuration with \(m\) \texttt{green} agents to a configuration with \(2m\) \texttt{green} agents for $m<n^{1/k}\log n$ is 
\[
(c+k)\frac{n^{1-1/k}(\log n)^{1/k}}{\sqrt{m}}
\]
with high probability (whp), for some constant $c>0$.
\end{lemma}

\begin{proof}
Assume that the number of \texttt{green} agents is exactly $m$. Consider \(k-1\) successive periods of length \(n^{1-1/k}(\log n)^{1/k}/\sqrt{m} \)
and an additional period of length \(cn^{1-1/k}(\log n)^{1/k}/\sqrt{m} \). We show that after all these periods, we obtain at least \(m\) new green agents whp.
 
We prove by induction that after the first $i$ periods 
there are at least $n_i$ agents that had at least $i$ contacts (interactions with \texttt{green} agents) for $i=0,1,2,\ldots,k-1$, where $n_0=n-m$ and for $i>0$,
\[n_i=m(m-1)(m-1)\cdots(m-i+1)\cdot n \left(\frac{\log n}{n}\right)^{i/k}(\sqrt{m})^{-i/2}.
\]

Now we prove the inductive step.
Assume that after initial $i-1$ periods there are at least $n_{i-1}$
agents with at least $i-1$ contacts.
Let \(X_t\) be a random variable that equals 1 if, at time \(t\), an agent with $i-1$ contacts experiences a new contact (with its $i$-th \texttt{green} agent), and 0 otherwise. 
If in time $t$ less than $n_i$ agents had $i$ contacts 
\[\Pr(X_t = 1) >  2(m-i)\frac{n_{i-1}-n_i}{n(n-1)}>1.5(m-i)\frac{n_{i-1}}{n}.\]
After the $i$-th period of length $n^{1-1/k}(\log n)^{1/k}/\sqrt{m}$ the expected value of $EX=\sum_t EX_t$ is at least 
\[1.5(m-i)\frac{n_{i-1}}{n}\cdot n\cdot \frac{n^{1-1/k}(\log n)^{1/k}}{\sqrt{m}}> 1.4 n_i  .\]
By Lemma \ref{czernow} this number is at least $n_i$ whp.
Thus also the number of
agents that had at least $i$ contacts is at least $n_i$ whp.

After $k-1$ periods of length $n^{1-1/k}(\log n)^{1/k}/\sqrt{m}$ there are at least $n_{k-1}$
agents that had at least $k-1$ contacts.
Now let us add an additional period of length $cn^{1-1/k}(\log n)^{1/k}/\sqrt{m}$.
We show that after this period we will have at least $m$ new green agents whp.
Let \(X_t\) be a random variable that equals 1 if, at time \(t\), an agent with $k-1$ contacts experiences a new contact (with its $k$-th \texttt{green} agent), and 0 otherwise.
Note that each time $X_t=1$, a new \texttt{green} agent is produced.
As long as less than $m$ agents became green, 
\[\Pr(X_t = 1) >  2(m-k+1)\frac{n_{k-1}-m}{n(n-1)}>\frac{n_{k-1}}{n}.\]
After one extra period of length $cn^{1-1/k}(\log n)^{1/k}/\sqrt{m}$ the expected value $EX=\sum_t EX_t$ is at least 
$cm(m-1)\cdots (m-k+1)\cdot m^{-k/2}\log n>cm\log n$.
By Lemma \ref{czernow} this number is at least $m$ whp.
Thus also the number of newly generated \texttt{green} agents is at least $m$ whp.
\end{proof}

\begin{lemma}\label{finish_k}
Starting with at least $n^{1/k}\log n$ \texttt{green} agents guarantees recolouring all $n$ agents to \texttt{green} in time 
$O(n^{1-1/k})$. 
\end{lemma}

\begin{proof}
If there are altogether $n^{1/k}\log n$ \texttt{green} agents, 
then for any \texttt{blue} agent one can define a random variable $X_t$ 
equal 1 if in time $t$ this agent interacts with a new \texttt{green} agent, and 0 otherwise.
The probability that a \texttt{blue} agent does not become \texttt{green} is $Pr(X_t=1)>0.9n^{1/k}\log n$.
In time $ckn^{1-1/k}$ the value $EX=\sum EX_t>0.9c\log n$.
By Lemma \ref{czernow} $EX>k$ whp, i.e.,\ each agent becomes \texttt{green} whp.
\end{proof}

\begin{theorem}\label{k-contact}
The stabilisation time of \(k\)-contact epidemic is \(O(n^{1 - 1/k} \log^{1/k} n)\) whp.
\end{theorem}

\begin{proof}
The execution time of \(k\)-contact epidemic is the sum of the times to increase the number of \texttt{green} agents from \(m = k\) to \(m = n\), and can be calculated using Lemmas
\ref{double_k} and \ref{finish_k}
\[
T = O(n^{1-1/k}) + \sum_{m=k,2k,4k,8k,\ldots, n^{1/k}\log n} (c + k) \frac{n^{1 - 1/k}(\log n)^{1/k}}{\sqrt{m}} = O(n^{1 - 1/k} \log^{1/k} n). 
\]
\end{proof}

\subsection{Localisation via multi-contact epidemic}
The localisation protocol presented in this section consists of two parts (see the formal description of Algorithm~\ref{alg:Pos_k} on page~\pageref{alg:Pos_k}).
In the first part of the protocol, the labels of $k+1$ agents (including the leader) become stable (positions of these agents become fixed) and these agents become \texttt{green}.
The counter $i,$ initially set to $1,$ of agents with stable (\texttt{green}) labels is maintained by the leader.
In the second part, the labels of all remaining (\texttt{blue}) agents become stable. And this is done by contacting with $k+1$ different \texttt{green} agents.
And once the agent is positioned, it becomes \texttt{green}.
We refer to this ({\em multilateration}) process as $(k+1)$-contact epidemic, see page \pageref{k-con}.

The positioning of each of the first $k$ \texttt{green} agents requires approval from the leader, who maintains a list of currently registered \texttt{green} agents.
More precisely, after an aspiring to be \texttt{green} agent $A_v$ interaction with all $i<k$ \texttt{green} agents is concluded, to become \texttt{green} $A_v$ must meet the leader and verify its list of \texttt{green} contacts to get approved. And when this happens, the leader updates the counter of \texttt{green} agents, and the new \texttt{green} agent $A_v$ is ready to calculate its projection onto the subspace spanned by its $i$ \texttt{green} predecessors and the leader, as well as its Euclidean distance from this subspace. Namely, the first $i$ coordinates of $A_v$'s label are determined by this projection, and the $(i+1)$-th coordinate (in newly formed dimension) is equal to its distance from the aforementioned subspace.
When positioning the remaining agents, we use the fact that interactions with $k+1$ \texttt{green} agents allow for the unambiguous determination of an agent's position. For detail check Algorithm~\ref{alg:Pos_k}.

%\input{m-algPosK}
% here we rewrite alg:Pos_k algorithm 
% we use lists of met green agents of size at most k+1
% where the k is the dimension of the space

\begin{algorithm}[tb]
\caption{Positioning in $k$ dimensions}
\label{alg:Pos_k}

\tcc{{Code executed by $\agentsymbol{u}$ exchanging data from 
 ${\agentsymbol{v}}$ during interaction}}
\Local{$\agentlabel{}$ -- position, $C$ -- colour, $L$ -- list of green agents' (position,distance)}
\Input{$L(A_v), C(A_v), \agentlabel{}(A_v),\distance{u}{v}$}
\Begin{
\tcc{Initial values}
    $L(A_u) \algoassign \texttt{empty list}$ \\
\eIf {$u$ is the leader} 
    {$C(A_u) \algoassign \texttt{green}; x(A_u) \algoassign (0,\dots,0)\ $} 
    {$C(A_u) \algoassign \texttt{blue}$}
\While{not all agents are positioned}{
    \If {$A_u \textrm{ is the leader}  \wedge C(A_v)=\texttt{blue} \wedge |L(A_u)| = |L(A_v)| \le k$}
    {append $(L(A_u),(x(A_v), \distance{u}{v}))$}
    \If{$A_v \textrm{ is the leader}  \wedge C(A_u)=\texttt{blue}
        \wedge |L(A_u)| = |L(A_v)| \leq k$}
    {$C(A_u) \algoassign \texttt{green}$\\
    $x(A_u) \algoassign$ a position consistent with $L(A_u)$ and $\big(x(A_v), \distance{u}{v}\big)$
    }
    \If{$C(A_u) \textrm{ is \texttt{blue}} \wedge C(A_v) 
    = \texttt{green} \wedge (A_v, \distance{u}{v}) \notin L(A_u)$}
    {
        append ($L(A_u), (x(A_v), \distance{u}{v})$)\\
        \If{$|L(A_u)| = k+1$}
        {
            $C(A_u) \algoassign green$ \\
            $x(A_u) \algoassign$ the exact position calculated using $L(A_u)$
        }
    }
}
%    $C[leader] \algoassign \texttt{green}$\\
}
\end{algorithm}

\begin{theorem}\label{t:PosK}
Algorithm \ref{alg:Pos_k} stabilises silently 
%labels of all agents in $k$ dimensions 
in time $O(n (\log n/n)^{1/(k+1)})$ whp.
\end{theorem}

\begin{proof}
As mentioned earlier Algorithm~\ref{alg:Pos_k} operates in two parts. In the first part, the protocol stabilises the labels of the leader and $k$ extra agents, creating an initial set of \( k+1 \) \texttt{green} agents in time $O(n (\log n/n)^{1/k})$ whp (Lemma~\ref{ini_k}).
The second part stabilises labels of all remaining agents via $(k+1)$-contact epidemic in time $O(n (\log n/n)^{1/(k+1)}$ 
(Theorem \ref{k-contact}).
\end{proof}

First we formulate the following lemma.

\begin{lemma}\label{czernow}
Consider $X=X_1+X_2+\cdots+X_n$ of $n$ independent $0-1$ random variables and any $\delta>0$. If the expected value $EX\ge c\log n,$ for $c>0$ large enough, then
$|X-EX|<\delta EX$ holds whp.
\end{lemma}

\begin{proof}
The following equality holds
\[
\Pr(|X-EX|<\delta EX)=\Pr(X>(1+\delta)EX)+\Pr(X<(1-\delta)EX)
\]
By Chernoff inequalities, for $c$ large enough and any parameter $\eta$, we get 
\[
\Pr(X>(1+\delta)EX)<e^{-EX\delta^2/(2+\delta)}<e^{-c\delta\log n}<n^{-\eta}/2
\]
and
\[
\Pr(X<(1-\delta)EX)<e^{-EX\delta^2/2}<e^{-c\delta^2\log n/2}<n^{-\eta}/2.
\]
Thus $\Pr(|X-EX|<\delta EX)<n^{-\eta}$.
\end{proof}

Now we formulate a lemma determining the time needed to position $i$-th agent
during the process of positioning the first $k$ green agents.

\begin{lemma}\label{ini_k}
Algorithm~\ref{alg:Pos_k} (first part), the $i$-th \texttt{green} agent is positioned in parallel time $O(n^{1-1/i}(\log n)^{1/i})$ whp.
\end{lemma}

\begin{proof}
Consider \(i-1\) successive periods of length \(n^{1-1/i}(\log n)^{1/i} \)
and an additional period of length \(cn^{1-1/i}(\log n)^{1/i} \). We show that after all these periods, whp a new green agent is positioned.
   
We now prove by induction that after the first $j$ periods 
there are at least $n_j$ agents that had at least $j$ interactions with green agents for $j=0,1,2,\ldots,i-1$, where $n_0=n-1$ and for $j>0$,
\(n_j=n \left(\frac{\log n}{n}\right)^{j/i}.
\)

We start with the inductive step.
Assume that after initial $j-1$ periods there are at least $n_{j-1}$
agents with at least $j-1$ {\em contacts}, i.e., interactions with at least $j-1$ \texttt{green} agents.
Let \(X_t\) be a random variable that equals 1 if, at time \(t\), an agent with $j-1$ contacts  has a new contact with its $j$-th \texttt{green} agent, and 0 otherwise. 
If in time $t$ less than $n_j$ agents had $j$ contacts then 
\[\Pr(X_t = 1) >  2\frac{n_{j-1}-n_j}{n(n-1)}>1.5\frac{n_{j-1}}{n}.\]
After the $j$-th period of length $n^{1-1/k}(\log n)^{1/k}$ the expected value of $EX=\sum_t EX_t$ is at least 
\[1.5\frac{n_{j-1}}{n}\cdot n\cdot {n^{1-1/i}(\log n)^{1/i}}> 1.4 n_j .\]
By Lemma~\ref{czernow} this number is at least $n_j$ whp.
Thus also the number of
agents that had interactions with at least $j$ green agents is at least $n_j$ whp.

Furthermore, after $i-1$ periods there are at least $n_{i-1}=(n/\log n)^{1/i}\log n$ agents
that experienced $i-1$ contacts.
Consider an extra period of length \(cn^{1-1/i}(\log n)^{1/i} \).
Let \(X_t\) be a random variable that equals 1 if, at time \(t\), an agent with $i-1$ contacts has a contact with the leader, and 0 otherwise. 
If in time $t$ none of these agents had interaction with the leader yet, we get 
\[\Pr(X_t = 1) >  2\frac{n_{i-1}}{n(n-1)}.\]
After an extra period of length $n^{1-1/k}(\log n)^{1/k}$ the expected value of $EX=\sum_t EX_t$ is at least 
\[2\frac{n_{j-1}}{n}\cdot n\cdot c{n^{1-1/i}(\log n)^{1/i}}= 2c\log n  .\]
And by Lemma \ref{czernow} this number is at least $1$ whp.
\end{proof}

\subsection{Faster localisation algorithm}

In this section, we demonstrate how to enhance the performance of Algorithm~\ref{alg:Pos_k} on the line, i.e., within a linear space $S.$ While we focus on one dimension, our observations can be utilised also in higher dimensions. In what follows, we propose an alternative 
%Algorithm \ref{alg:Improved}
localisation protocol which positions agents not only by using 2-contact epidemic but also utilising interactions between agents with a single successful contact coloured later \texttt{greenish}.
In the new algorithm we take advantage of the fact that there are only two types of
interactions between \texttt{greenish} agents
not leading to positioning of both of them.
The first type refers to interactions between agents that became greenish via contacting the same \texttt{green} agent.
In the second, each of the two interacting \texttt{greenish} agents is at the same distance and on the same side of their unique \texttt{green} contact.
We show that these types of interactions do not contribute significantly to the total time of the solution.

\begin{lemma}\label{double_i}
The time needed to increase the number of \texttt{green} agents 
from $m$ to $2m,$ for $m\in [2,n^{0.9}],$ is $O((n\log n/m)^{1/3})$ whp. 
\end{lemma}

\begin{proof}
First we consider an initial period $I_1$ of length $(n\log n/m)^{1/3}$.
The average number of \texttt{greenish} agents produced by $m$ \texttt{green} agents in time $(n\log n/m)^{1/3}$ is $(m^2n\log n)^{1/3}$.
By Lemma \ref{czernow} this number is at least $0.9(m^2n\log n)^{1/3}$ whp.
For any given \texttt{green} agent $A_u$, one can define a subset ${\cal S}_u$ of \texttt{greenish} agents originating from contact with agents other than $A_u$.
The cardinality of $S_u$ is on average at least $0.9(m-1)(n\log n/m)^{1/3}.$
%For any given \texttt{green} agent $A_u$, this set contains a subset ${\cal S}_u$ of on average at least $0.9(m-1)(n\log n/m)^{1/3}$ agents originating from contact with agents other than $A_u$.
By Lemma~\ref{czernow} this number is greater than $0.8(m-1)(n\log n/m)^{1/3}$ whp.
For a given agent $A_v$, which turned \texttt{greenish} after contacting \( A_u \), there are fewer than \( 2m \) \texttt{greenish} agents in ${\cal S}_u$ with whom no interaction leads to the positioning of \( A_v \), unless the number of \texttt{green} agents exceeds $2m$ (they are the second kind of non-positioning agents).
Let us consider any set $Z$ of at most $2m$ agents located at points
that are translations of the positions of \texttt{green} agents by a fixed vector.
The expected number of \texttt{greenish} agents belonging to $Z$ is at most
$n\cdot (n\log n/m)^{1/3}\cdot (2m)^2/n^2=4(m^2n\log n)^{1/3}m/n$.
So by Chernoff bound this number is at most $0.1(m^2n\log n)^{-1/3}$ whp.

Thus for a given \texttt{greenish} agent $A_v,$ the number of other \texttt{greenish} agents with whom interactions position $A_v$ is whp at least
\[
0.8(m-1)\left(\frac{n\log n}{m}\right)^{1/3}-0.1(m^2n\log n)^{1/3}> 0.6(m^2n\log n)^{1/3}.
\]
Let $I_2$ be a time interval of length  $c(n\log n/m)^{1/3}$
that follows immediately after $I_1$.
The probability that an interaction $t$ between two \texttt{greenish} agents is the one  that positions them and makes them \texttt{green} is at least $0.6(m^2 n^{-2}\log n)^{2/3}$.
An average number of interactions that position pairs of \texttt{greenish} agents in period $I_2$ is $0.6c m\log n$. By Lemma \ref{czernow} this yields at least $m$ new \texttt{green} agents whp.
\end{proof}

\begin{lemma}\label{finish_i}
The time in which the number of green agents increases
from $n^{0.9}$ to $n$ is $O(n^{0.1}\log n)$. 
\end{lemma}

\begin{proof}
Consider all interaction of agent $A_u$ during period $cn^{0.1}\log n$.
The probability of having at most one interaction with a green agent for $A_u$ is
\begin{multline}
\left(1-\frac{2n^{0.9}}{n^2}\right)^{cn^{1.1}\log n}+
  cn^{1.1}(\log n){0.9}{n^2}\left(1-\frac{2n^{0.9}}{n^2}\right)^{cn^{1.1}\log n-1} < \\
  <c(\log n) e^{c\log n}<e^{2c\log n}=n^{2c\ln 2}.
\end{multline}
So, there is $c>0,$ s.t., all agents have at least two interactions with green agents during the considered period $cn^{0.1}\log n$ whp.
\end{proof}

\begin{theorem}\label{t:improved}
The stabilisation time of the improved algorithm  is \(O((n\log n)^{1/3})\) whp.
\end{theorem}

\begin{proof}
The execution time of the improved algorithm is the sum of the chunks of time needed to increase the number of \texttt{green} agents from \(m = 2\) to \(m = n\), and it can be calculated using Lemmas
\ref{double_i} and \ref{finish_i}
\[
T = O(n^{0.1}\log n) + \sum_{m=2,4,8,\ldots, n^{0.9}} O((n\log n/m)^{1/3})
= O((n\log n)^{1/3}). 
\]
\end{proof}

\section{Self-stabilising localisation}
\label{s:self}

We begin with a brief description of the solution.
The proposed self-stabilising protocol in $k$-dimensions is based on iterative {\em rounds}, with each round utilising three mechanisms:
{\em leader election} incorporated into {\em leader-based localisation} (the first stage of a round),
and a {\em buffering mechanism} (the second stage of a round) adopted from~\cite{GGS25} (Section~5.2). 

During leader election, each agent draws $\log n$ random bits, and if none of these bits is $0$, the agent proceeds to the actual localisation protocol as a leader. 
Note that after all agents complete drawing random bits, which takes $O(\log n)$ time,
a unique leader is elected with constant probability, see Lemma~\ref{l:leader}.
Regardless of whether a unique leader is elected, the leader-based localisation protocol continues (possibly indefinitely) unless an {\em anomaly} is detected. 
%i.e., state inconsistency, is detected, triggering reset signla preceded by a buffering stage), or localisation completes (achieving silent stabilisation). During leader-based localisation, each agent tracks its interactions using a countdown counter, initially set to the expected stabilisation time of the protocol.
%
Two types of anomalies may arise. 
The first, {\em label inconsistency}, occurs only between two \texttt{green} agents. This kind of anomaly occurs when either the initial configuration has conflicting labels of \texttt{green} agents or multiple leaders were elected during the current round.
When all agents are \texttt{green} this type of anomaly is detected in time $O(k\log n)$, see Lemma~\ref{l:rigid}.
The second anomaly occurs when any non-\texttt{green} agent’s counter reaches its deadline $O(n(\log n/n)^{1/(k+1)})$, i.e., the expected stabilisation time of Algorithm~\ref{alg:Pos_k}), indicating that the localisation process is not completed on time.
Upon detecting any anomaly, a reset signal is initiated and propagated via a simple epidemic in $O(\log n)$ time. This together with the (collection and) buffering mechanism adopted from~\cite{GGS25} (Section 5.2) ensures sufficient time, namely $O(\log n)$, to reset states of all agents, and prepare them for the next round.

Finally, the time of each round is dominated by the expected time of  localisation stabilisation $O(n(\log n/n)^{1/(k+1)}).$
Also, the probability of stabilisation in a round is constant, primarily driven by the constant probability associated with successful leader election.
Thus, after $O(\log n)$ rounds and a total time of $O(n(\log n/n)^{1/(k+1)}\log n)$, the localisation protocol self-stabilises with high probability (whp), see Theorem~\ref{t:self}.

%In this version, we focus on establishing the existence of a relatively fast self-stabilising protocol, with stabilisation time comparable to the algorithms presented in Section~\ref{s:leader}. 
Our result demonstrates that self-stabilisation is not only feasible but also efficient, i.e., comparable to Algorithm~\ref{alg:Pos_k} presented in Section~\ref{s:leader}. 
%when agents have access to an exact value of $\log n$. Moreover, if only an approximate value of $\log n$ is known, the protocol still stabilises, albeit with a slower convergence time.
%
Moreover, although our self-stabilising localisation protocol relies on leader election, it only requires knowledge of $\log n$, whereas self-stabilising leader election requires precise knowledge of $n$, see~\cite{CIW12}. This is not a contradiction because we only need leader election to succeed with constant probability, which together with efficient anomaly detection, discussed in Section~\ref{s:augment}, and lightly modified buffering mechanism from~\cite{GGS25}, discussed in Section~\ref{s:buffer}, allows us to restart leader election and in turn the localisation protocol, when necessary.

%\subsection{Further detail and proofs}

Given that the self-stabilising localisation protocol combines multiple sub-protocols, we provide further details below, including the relevant proofs.
Note that, as this version emphasises the feasibility of self-stabilisation, the optimisation of both time and space for the protocol will be addressed in the full version of this paper.

\subsection{Memory utilisation}
We begin by analysing the memory utilisation of the protocol. The number of states required for leader election and buffering is $O(\log n)$, as shown in Lemma~\ref{l:leader} and Lemma 21 of~\cite{GGS25}, respectively. However, agents occupying states used for leader election and buffering do not participate in the actual localisation process. As a result, the overall state space of the self-stabilising protocol is dominated by that of Algorithm~\ref{alg:Pos_k}, extended by a deadline counter of size $O\left(n(\log n / n)^{1/(k+1)}\right)$ used by the \texttt{blue} agents.

\subsection{Leader election\label{s:leaderelection}}

We begin by presenting a suitable leader election protocol that operates in $O(\log n)$ time and uses $O(\log n)$ space. This protocol is executed at the beginning of the first stage of each round and is followed by the main localisation protocol, Algorithm~\ref{alg:Pos_k}, augmented with anomaly detection and deadline counters.

%We need the following lemma in which 
%leader election protocol utilising $O(\log n)$ time and $O(\log n)$ space is presented and analysed.

Leader election is done by independently assigning a leader role to each agent with probability $p=\Theta(1/n)$.
In this way, the probability of electing exactly one leader is constant and is maximised when $p=1/n$ and is then approximately $1/e$.
In our protocol, each agent independently tosses a symmetric coin $\log n$ times and becomes the leader if it gets heads on all tosses.
We are left to describe the protocol implementing this process.

States in this protocol are denoted as pairs $(A, i)$, in which $A \in \{N, H, T\}$ and $i \in \{1, \ldots, \log n, L, F\}$. The first component, $A$, represents the type of coin toss state: {\em Neutral} ($N$), {\em Head} ($H$), or {\em Tails} ($T$). The second component, $i$, either serves as a counter indicating progress toward successful leader election, or denotes a status: $L$ for leader and $F$ for follower.
On leaving the buffering mechanism (the second stage of the previous round)
%i.e. immediately after leaving the state W, 
all agents receive state $(N,1)$, and the only meaningful interactions in the leader election protocol are:
\begin{itemize}
\item
The creation of $H:T$ state pairs ensuring that there are consistently the same number of agents with $H$ and $T$ states in the population: 
$(N,*)+(N,*)\rightarrow (H,*)+(T,*).$
\item
A coin toss in which the initiator gets heads:
$(*,i)+(H,*)\rightarrow (*,i+1)+(H,*),$ if $i<\log n$,
and $(*,\log n)+(H,*)\rightarrow (*,L)+(H,*)$, otherwise.
\item
A coin toss in which the initiator gets tails:
$(*,i)+(T,*)\rightarrow (*,F)+(T,*)$, if $i\in\{1,\ldots,\log n\}$.
\end{itemize}

\begin{lemma}\label{l:leader}
 A unique leader is elected in $O(\log n)$ time with a constant probability.
\end{lemma}
\begin{proof}
Let us assume that at time $0$ all agents have already left the buffer and joined the leader election protocol. For as long as the number of agents in $(N,*)$ states exceeds $n/2$ the probability of forming a $H:T$ pair in a given interaction exceeds $1/4$. Thus, the expected time in which the number of agents in $(N,*)$ states falls below $n/2$ is less than $2$, coinciding with $2n$ interactions. And whp this time is less than 3.
After this time, the probability that an agent with a counter in an interaction makes a coin flip is $1/2n$. Thus, in a time period of $2c\cdot \log n$ it performs at least $c\cdot \log n$ coin tosses on average.
From Chernoff's inequality for a sufficiently large constant $c$ it makes at least $c\cdot \log n$ coin tosses whp. From the union bound in total time $O(\log n)$ all agents will perform $\log n$ coin tosses each, which will establish their leader ($L$) or follower ($F$) status.
\end{proof}

We still need to explain how to adopt this leader election protocol in the self-stabilising leader-based localisation protocol.
Later in Section~\ref{s:buffer}, we show that at some point all agents are located in the buffer, and ready to proceed to the next round. When the first agent leaves the buffer, it assumes the state $(N,1)$, which triggers an epidemic process that informs all other agents in the buffer to adopt the same state $(N,1)$. Note that all agents participating in the leader election protocol contribute to this epidemic, ensuring that the entire population engages in leader election within $O(\log n)$ time.
Upon completing the leader election, agents with confirmed leader or follower status start leader-based localisation protocol adopting \texttt{green} and \texttt{blue} states, respectively.
In addition, they never give up their coin ($H|T|N$) attributes to guarantee fairness of the remaining coin tosses for agents still executing leader election protocol.
%use an additional counter of size $O(\log n)$ to delay the execution of the leader-based localisation protocol until the status of all agents is determined whp. When the first agent enters the leader-based localisation phase, it triggers an epidemic process that instructs all remaining agents still engaged in leader election to adopt the \texttt{green} state if they are leaders, or the \texttt{blue} state otherwise, and begin contributing to localisation. 
%These two modifications require an additional $O(\log n)$ states and incur an additional $O(\log n)$ time overhead. 
Thus, after integration, the overall running time of the leader election protocol remains $O(\log n)$.
%, with a total state complexity of $O(\log n)$.

\subsection{Augmented leader-based localisation}\label{s:augment}

On the conclusion of leader election process leader-based localisation (Algorithm~\ref{alg:Pos_k}) is executed, however, with a few modifications. 
%As mentioned in Section~\ref{s:leader} all agents executing leader-based localisation protocol contribute to the epidemic instructing all agents still engaged in leader election to adopt \texttt{green} and \texttt{blue} states. In addition, 
Namely, the interaction counter with the deadline $O\left(n(\log n / n)^{1/(k+1)}\right)$ is added to every new \texttt{blue} agent. This is to ensure that the first stage does not exceed the expected stabilisation time of Algorithm~\ref{alg:Pos_k}. 
One also needs to manage anomaly detection. Recall, that two types of anomalies may arise. 
The first, {\em label inconsistency}, occurs only between two \texttt{green} agents. 
%This kind of anomaly occurs when either the initial configuration has conflicting labels of \texttt{green} agents or multiple leaders were elected during the current round.
%
In particular, when all agents become \texttt{green} this type of anomaly is detected in time $O(k\log n)$, see the lemma below.

\begin{lemma}\label{l:rigid}
    When all agents are \texttt{green}, the anomaly based on labels inconsistency is detected in time $O(k\log n)$ whp.
\end{lemma}
\begin{proof}
    It is known that a random graph with nodes embedded in $k$-dimensional space and $O(k \log n)$ edges is globally rigid~\cite{LNP+23}. This implies that there exists a unique embedding of the graph, up to isometries such as rotation, translation, and reflection. Consequently, if the labels of any \texttt{green} agents are inconsistent, such inconsistencies will be detected after $O(k \log n)$ interactions.
\end{proof}

The second anomaly occurs when any \texttt{blue} agent’s counter reaches the deadline.
%, i.e., the expected stabilisation time of Algorithm~\ref{alg:Pos_k}), indicating that the localisation process is not completed on time.
%
Upon detecting either of the anomalies, a reset signal is triggered initiating the buffering mechanism discussed in the next section.

\subsection{Buffering mechanism}\label{s:buffer}

The buffering for reset of the self-stabilising localisation protocol is solved in the same way as the buffering in the ranking protocol with $O(\log n)$ additional states from \cite{GGS25}. In this work, the signal for reset is the detection of an anomaly. It causes the state of the agents that performed the anomaly detection to be set to the first state of the buffer: $X_1$. The buffer is a line consisting of $2d\log n$ states: $X_1,X_2,\ldots, X_{2d\log n}$ for $d$ large enough. 
We assign \texttt{red} colour to the states $X_1,\ldots,X_{d\log n}$ and \texttt{white} (\texttt{green} in~\cite{GGS25}) to the remaining buffer states.

We define the following transitions for these states:
\begin{itemize}
\item {\em Progress} on the line: $X_i+X_j\rightarrow X_{i+1}+X_{i+1}$ for $i\le j$.
\item {\em Reset propagation} by \texttt{red} agents: $X_i+A\rightarrow X_1+X_1$, if $i\le d\log n$ and $A$  not on the line.
\item {\em Buffer departure} by \texttt{white} agents: $X_{2d\log n}+X_{2d\log n}\rightarrow (N,1)+(N,1)$, where $(N,1)$ is the initial state of the leader election phase. 
Also $X_i+A\rightarrow (N,1)+A$, where $i> d\log n$ and $A$ is not on the line.
\end{itemize}

The analysis of the buffering process is summed up in the paper \cite{GGS25} by the following Lemma, proved there

\begin{lemma}[Lemma 21 in paper \cite{GGS25}]\label{l:upper-line}
    There exists $d'>0$ s.t. for any $d>0$ there is $c>0$
    for which after $c\log n$ time since state $X_1$ (reset signal) arrival, all agents are in line states $X_i$ with indices $d\log n< i\le (d+d')\log n$ whp.
\end{lemma}

If we substitute our $d$ for $d$ and $d'$ in this lemma we obtain that at time $O(\log n)$ after the first agents appear in state $X_1$, all agents are \texttt{white}.
On the other hand, when we substitute $2d$ for $d$ in the lemma and $d$ for $d'$, we obtain that at time $O(\log n)$ the agents start to leave the line of the buffer.

We conclude Section~\ref{s:self} with the following theorem.

\begin{theorem}\label{t:self}
The localisation protocol presented in this section self-stabilises silently in time $O(n(\log n/n)^{1/(k+1)}\log n)$ whp.
\end{theorem}

\iffalse
\begin{theorem}
The algorithm presented above self-stabilises in time~\(O(n(\log n/n)^{1/(k+1)}\log n)\) whp.
\end{theorem}

- red, black, blue, green (countdown clock zeroed)
- counters: pipeline + leader election + countdown clock

- memory usage 

k+1 labels, counters, colour, distance

anomalies detection
- two inconsistent green
- inconsistent green abd blue
- blue with zeroed countdown clock $O(n\log n)$ whp

- structure of the solution

blocks

(1) leader election (blue -> green)

local counting of 1s in time O(log n) 

(2) labelling (green + blue)

time $O(n(\log n/n)^{1/(k+1)})$ - controlled by countdown clock

(3) pipelining (red + black possibly interacting with blue and green)

process used in other paper - time $O(\log n)$

The protocol run in rounds.
Each round operates in time $O(n(\log n/n)^{1/(k+1)})$
The success of each round is constant.
After $O(\log n)$ rounds localisation stabilises whp

Conclusive remarks including more arbitrary distribution of points. 
\fi

% here the stuff with the sense of direction

\section{Self-stabilising localisation in vector query model\label{s:vector}}

In this section, we adopt the {\em vector query model}, assuming first that \( S \) is a linear space where agents \( \agentsymbol{0}, \dots, \agentsymbol{n-1} \) are arbitrarily positioned at (unknown to the agents) points \( \agentposition{0}, \dots, \agentposition{n-1} \), respectively. 
We use notation $\agentvector{i}{j}$ to denote the vector connecting $\agentposition{i}$ and $\agentposition{j},$ in this case $\agentvector{i}{i}=\agentposition{j}-\agentposition{i}.$
Additionally, each agent \( \agentsymbol{i} \) possesses a {\em hypothetical coordinate} referred to as its {\em label} \( \agentlabel{i} \).
In this variant of population protocols, during an interaction of \( \agentsymbol{i} \) with \( \agentsymbol{r} \),  the initiator
$\agentsymbol{i}$ learns about vector \( \agentvector{\initiator}{\responder} \) and label $x_r$.
We show that this model is very powerful as it allows design of an optimal $O(\log n)$ time self-stabilising labelling protocol. 
\begin{algorithm}
\caption{Fast positioning in one dimension}
\label{alg:1Dexact}
\tcc{\commontcc}
\Local{$\agentlabel{\initiator}$}
\Input{$\agentlabel{\responder}, \agentvector{\initiator}{\responder}$}
\Begin{
%    If we assume the vector is given instead of 
%    the distance and the direction we can simplify here: 
%    \eIf{$\signaldirection == 1$} 
    {$\agentlabel{\initiator} \algoassign 
        \max\{\agentlabel{\initiator}, \, \agentlabel{\responder} - \agentvector{\initiator}{\responder}\}$
    }
%    {
%        $\agentlabel{\initiator} \algoassign 
%        \max\{\agentlabel{\initiator}, \, \agentlabel{\responder} + %\distance{\initiator}{\responder}\} }    
}
\end{algorithm}

We start with a trivial fact concerning label updates.

\begin{fact}
\label{fact:never_decreases}
For each agent $A_i$ its label $\agentlabel{i}$ never decreases during the execution of 
Algorithm~{\ref{alg:1Dexact}}.
\end{fact}

Let $M = \max\limits_{0\le i \le n-1}\big (\agentlabel{i} - \agentposition{i}).$ 
%be the maximum difference between the position of a robot and its current coordinate. 

\begin{fact} 
\label{fact:M_is_const}
The value of $M$ does not change during the execution of 
Algorithm~{\ref{alg:1Dexact}}.
\end{fact}
\begin{proof}
Consider an effective update of label $\agentlabel{\initiator}$, where $\agentlabel{\initiator}'$ is the updated value of this label.

$\agentlabel{\initiator}' = \agentlabel{\responder} - \agentvector{\initiator}{\responder} 
= \agentlabel{\responder} - (\agentposition{\responder} - \agentposition{\initiator}) 
$. Thus,  
$\agentlabel{\initiator}' - \agentposition{\initiator}  = 
 \agentlabel{\responder} - \agentposition{\responder}  \leq M.$
\end{proof}

Thus after such effective update
 $\agentsymbol{\initiator}$ adopts new label $\agentlabel{\initiator}=\agentlabel{\responder} - \agentposition{\responder}$.
Let $S_M= \{ A_i :  \agentlabel{\initiator} - \agentposition{\initiator}  = M\}$. 

\begin{fact}
\label{fact:labels_consistent_in_MS}
For any two agents $A_i,A_j \in S_M$ we have 
$\agentlabel{j} - \agentlabel{i} = \agentposition{j} - \agentposition{i}$.
\end{fact}

\begin{proof}
Note that 
$\agentvector{i}{j} = \agentposition{j} - \agentposition{i} = 
\agentlabel{j} - M - (\agentlabel{i} - M) 
= \agentlabel{j} - \agentlabel{i}$.
\end{proof}

In consequence, the labels of any two agents in $S_M$ are consistent (with respect to their real positions).

\begin{fact}
\label{fact:epidemic}
After $\agentsymbol{\initiator} \notin S_M$ 
interaction with $\agentsymbol{\responder}\in S_M$, $\agentsymbol{\initiator}$ becomes an element of $S_M$.
\end{fact}

\begin{proof}
As $\agentsymbol{\initiator} \notin S_M,$ we get 
$\agentlabel{\initiator} - \agentposition{\initiator} < M$.
However, as $\responder \in S_M$ we also get  
$\agentlabel{\initiator} - \agentposition{\initiator} <
 \agentlabel{\responder} - \agentposition{\responder},$ and 
in turn
$\agentlabel{\initiator} < 
    \agentlabel{\responder} - \agentposition{\responder} + \agentposition{\initiator} < 
    \agentlabel{\responder} - \agentvector{\initiator}{\responder}$. 
Thus $\agentlabel{\initiator}$
becomes
$\agentlabel{\responder} - \agentvector{\initiator}{\responder}$, which is equal to $M$ as $\responder \in S_M$.
\end{proof}

%\lkcomment{I guess that such an epidemic algorithm was analyzed in the literature and we can say that Algorithm 1 is $O(n \log n)$ steps,
% Leszek/Grzegorz, could you please add an appropriate citation?}

\begin{theorem}
The localisation Algorithm~\ref{alg:1Dexact} self-stabilises silently in the optimal time $O(\log n)$ whp.
\end{theorem}

% Each agent's label uniquely determines the origin of a coordinate system. During interactions, % both participants receive labels related to the coordinate system with the origin at the point %that is the minimum of both previous origins. At the end of the protocol, all agents have the %origin of their coordinate systems at a common point, which is the minimum of all origins. 
% Since this minimum is propagated by an epidemic, the protocol has a time complexity of $O(\log n)$.
%\lkcomment{I have commented out the text above as in fact the maximum value of a difference is propagated. We still need an appropriate citation}
\begin{proof}
The membership of agents in $S_M$ is spread via one-way epidemic in time $O(\log n)$ whp.
This protocol is silent as after all agents are included in $S_M,$ further label updates are not possible. Ultimately, this protocol is self-stabilizing, requiring no initial assumptions about agents' labels, and achieves time optimality of \( O(\log n) \), corresponding to the communication bound in population protocols.
 %During all interactions, participants are included into the set $S_M$ by an epidemic see 
 %Facts~\ref{fact:never_decreases}, \ref{fact:M_is_const}, \ref{fact:labels_consistent_in_MS}, and \ref{fact:epidemic}, 
 %the protocol has probabilistic parallel time complexity of $O(\log n)$.
\end{proof}

Finally, we observe that Algorithm~\ref{alg:1Dexact} can be extended to higher (fixed) dimensions. One can apply multiple instances of one dimensional protocol on coordinates in each dimension, see Algorithm~\ref{alg:d-exact}
%In $d$-dimensional space we can apply the same technique for each dimension. Let the coordinate of agent $\initiator$ be a tuple $\agentvlabel{\initiator}$ of $d$ coordinates 
%$[\agentvlabel{\initiator}[1], 
%  \agentvlabel{\initiator}[2], 
%  \ldots \agentvlabel{\initiator}[d]]$.
\begin{algorithm}
\caption{Fast positioning in fixed $k$ dimensions}
\label{alg:d-exact}
\tcc{\commontcc}

\Local{$\agentlabel{\initiator}=[x_i[0],\dots,x_i[k-1]]$}
\Input{$\agentlabel{\responder}, \, \agentvector{\initiator}{\responder}$}
\Begin{
    \ForEach{$j \in [0, k-1]$}{
        {$\agentlabel{\initiator}[j] \algoassign 
        \max\{\agentlabel{\initiator}[j], \, \agentlabel{\responder}[j] - \agentvector{\initiator}{\responder}[j]\}$}
    }
}
\end{algorithm}

\begin{theorem}\label{t:super}
The localisation Algorithm~\ref{alg:d-exact} 
self-stabilises silently in a fixed $k$ dimensional space in the optimal time $O(\log n)$ whp.
\end{theorem}
\begin{proof}
By directly applying the union bound.  
\end{proof}
%\section{Open problems}
\section{Concluding remarks}

In this paper, we introduce a novel variant of spatial population protocols and explore its applicability to the distributed localization problem.
Any meaningful advances in this problem could pave the way for developing faster and more robust lightweight communication protocols suitable for real-world applications. It could also provide insights into the limitations of what can be achieved in such systems.

Several challenges remain unresolved in this work. Firstly, we do not account for inaccuracies in distance measurements. As no measuring device is perfect, future studies should model measurement errors, which can significantly affect system stability and performance. Our preliminary studies on the fast positioning protocol from Section~\ref{s:vector} reveal a phenomenon called label \emph{drifting} in the presence of errors, resembling phase clock behaviour. This drifting can be controlled to achieve near-complete label stabilisation and may support mobility coordination in large robotic agent populations.

Second, we leave the issue of limited communication range unaddressed. In real-world scenarios, agents often cannot communicate with all other agents in the system, which adds further complexity. As mentioned in the introduction, it is well known that limited communication range and arbitrary network topologies can lead to intractable localization problems in the worst case. Therefore, a promising direction for future research would be to focus on specific classes of network topologies for which lightweight localization protocols are more likely to be effective.

A third challenge is the development of efficient localisation algorithms for mobile agents. In this case, assumptions about the relative speeds of communication and movement would likely be necessary to ensure that data from previous positions can still be utilised effectively. 

Finally, of independent interest are further studies on the computational power of spatial population protocols, both in comparison to existing variants of population protocols and in relation to various geometric problems, types of queries, distance-based biased communication, and other related topics.
%\newpage
%Moreover, integrating data from an agent's gyroscope and its communication with other agents could help reduce errors in position estimation.
%In this paper we introduce a novel variant of {\em spatial} population protocols and we study suitability of this variant to distributed localisation problem.

%Solving this, might lead to very robust 
%algorithms implementable in real world applications
%or help determine what results are not possible to achieve.

%One of the problems not solved in this paper is the inaccuracy in distance measurement. 
%No measuring devices are perfect and measurement inaccuracies should be modelled in further considerations, as they have an impact on system stabilization.

%The second problem is the limited range of communication. In the real world, not all agents can communicate, so an interesting further direction could be to study the limited range or even arbitrary connection structure.

%Another challenge is to develop efficient localization algorithms for agents that move.
%Here, it would probably be necessary to make assumptions about the communication time relative to the movement speed so that the previous position data could be used.
%It would be interesting to couple the agent's gyroscope and communication with other agents to reduce the error in calculating the position.

\bibliography{ar-references}

\end{document}